\documentclass{svproc}
\usepackage{url}

\usepackage{amsmath} 
\usepackage{amsfonts} 
\usepackage{mathrsfs}

\usepackage{amsthm} 
\usepackage{thmtools}
\usepackage{amssymb} 
\usepackage{graphicx} 
\usepackage{enumerate} 
\usepackage{todonotes} 
\usepackage{xspace} 
\usepackage{tabularx}
\usepackage{extarrows}
\usepackage{soul}
\usepackage{subcaption}
\usepackage{siunitx}
\usepackage{wrapfig}
\usepackage[bookmarks=true,hidelinks]{hyperref}
\def\arxiv{} 

\let\emptyset\varnothing

\theoremstyle{definition} 
\newtheorem{x}{X} 
\newtheorem{y}{Y} 
\newtheorem{z}{Z} 
\newtheorem{definition}[x]{Definition} 
\newtheorem{problem}[x]{Problem} 
 
\newtheorem{lemma}[x]{Lemma} 
\newtheorem{theorem}[x]{Theorem} 
\newtheorem{example}[x]{Example}
\newtheorem{property}[x]{Property}
\newtheorem{idea}[y]{Idea}
\newtheorem{conjecture}[z]{Conjecture}

\providecommand{\Language}[1]{\ensuremath{\mathcal{L}(#1)}}

\newcommand{\cover}[1]{\ensuremath{\mathbf{#1}}}

\providecommand{\ssl}[1]{{\scriptsize\textsf{#1}}}

\newcommand*{\probleminternal}[4]{
	\par
	\medskip
	\noindent\fbox{\parbox{0.98\columnwidth}{
		\textbf{#4:} {#1} \\[0.05in]
		\renewcommand{\tabcolsep}{2pt}
		\begin{tabularx}{\linewidth}{rX}
			\emph{Input:} & #2 \\
			\emph{Output:} & #3
		\end{tabularx}
	}}
	\par
	\medskip
	\par
}

\let\emptyset\varnothing

\newcommand*{\ourproblem}[3]{\probleminternal{#1}{#2}{#3}{Problem}}
\providecommand{\defemp}[1]{\emph{#1}} 

\newcommand{\pfm}{{\sc pfm}\xspace}
\newcommand{\spfm}{{\sc so-fm}\xspace}
\newcommand{\mpfm}{{\sc mo-fm}\xspace}

\providecommand{\reachedvf}[3]{\ensuremath{\mathcal{V}_{#2}(#1, #3)}}
\providecommand{\reachedv}[2]{\ensuremath{\mathcal{V}(#1, #2)}}
\providecommand{\reachedc}[2]{\ensuremath{\mathcal{C}(#1, #2)}}

\providecommand{\reaching}[2]{\ensuremath{\mathcal{S}^{#1}_{#2}}}
\providecommand{\extensions}[2]{\ensuremath{\mathcal{L}_{#1}(#2)}}

\newcommand{\Aux}{\ensuremath{\textsc{Zip}}}

\providecommand{\compatibilitygraph}[1]{\ensuremath{\mathcal{K}(#1)}}
\providecommand{\compatible}{\ensuremath{\sim_{c}}}

\providecommand{\auxes}[1]{\ensuremath{\mathscr{Z}(#1)}}
\providecommand{\auxtext}{zipper constraint\xspace}
\providecommand{\auxestext}{zipper constraints\xspace}

\providecommand{\gmca}{{\sc gmczc}\xspace}
\providecommand{\mcca}{{\sc mcczc}\xspace}
\providecommand{\sde}[1]{\ensuremath{\operatorname{\textsc{Sde}}({#1})}}

\providecommand{\negation}[1]{\ensuremath{\overline{#1}}}

\providecommand{\proofinextension}{\textit{{The proof can be found in the
extended version~\cite{zhang2020Cover}.}}}

\newcommand\blfootnote[1]{%
  \begingroup
    \renewcommand{\thefootnote}{} 
    \footnotetext{#1}
    \renewcommand{\thefootnote}{\arabic{footnote}}
  \endgroup
}

\newcommand\shortenXor[2]{\ifdefined\arxiv #1\else #2\fi}

\newcommand\shortvspace[1]{\ifdefined\arxiv\else \vspace*{#1}\fi}
\newcounter{tecounter}
\newenvironment{tightenumerate}
{
    \begin{list}{\arabic{tecounter}\addtocounter{tecounter}{1})}{%
    \setcounter{tecounter}{1}
        \setlength{\leftmargin}{08pt}
        \setlength{\topsep}{1pt}
        \setlength{\partopsep}{0pt}
        \setlength{\itemsep}{2pt}
        \setlength\labelwidth{0pt}}
        \ignorespaces}
{\unskip\end{list}}

\DeclareFontFamily{U}{mathb}{\hyphenchar\font45}
\DeclareFontShape{U}{mathb}{m}{n}{<5> <6> <7> <8> <9> <10> gen * mathb
<10.95> mathb10 <12> <14.4> <17.28> <20.74> <24.88> mathb12}{}
\DeclareSymbolFont{mathb}{U}{mathb}{m}{n}
\DeclareMathSymbol{\rcirclearrow}{0}{mathb}{'367}

\begin{document}

\title{Cover Combinatorial Filters and their Minimization Problem \shortenXor{(Extended Version)}{}}

\author{Yulin Zhang \and Dylan A. Shell}
\institute{Department of Computer Science \& Engineering,\\ Texas A\&M University, College Station TX 77843, USA\\
\email{{yulinzhang}|{dshell}@tamu.edu}}
\ifdefined\arxiv
	\titlerunning{Cover Combinatorial Filters and their Minimization Problem}
\fi
\maketitle

\shortvspace{-24pt}
\begin{abstract}
Recent research has examined algorithms to minimize robots' resource footprints.
The class of combinatorial filters (discrete variants of widely-used
probabilistic estimators) has been studied and methods for reducing their space
requirements introduced.  This paper extends existing combinatorial filters by 
introducing a natural generalization\shortenXor{ that we dub}{:} cover combinatorial filters. 
In addressing the new\,---but still NP-complete---\,problem of minimization of
cover filters, \shortenXor{this paper shows}{we show} that multiple concepts previously believed\shortenXor{to be true}{} about combinatorial filters (and actually conjectured, claimed, or
assumed to be) are in fact false.  For instance, minimization does not induce an
equivalence relation.  We give an exact algorithm for the cover filter
minimization problem.  Unlike prior work (based on graph coloring)
we consider a type of clique-cover problem, involving a new conditional
constraint, from which we can find more general relations.  
In addition to solving the more general problem, the algorithm also corrects flaws present in all prior filter reduction methods.  
In employing SAT, the algorithm
provides a promising basis for future practical development.
\end{abstract}
\shortvspace{-14pt}

\blfootnote{This work was supported by the NSF through awards
 \href{http://nsf.gov/awardsearch/showAward?AWD_ID=1453652}{IIS-1453652} and
 \href{http://nsf.gov/awardsearch/showAward?AWD_ID=1527436}{IIS-1527436}.
 }

\section{Introduction}

As part of the long history of research in robotic minimalism, a recent
thread has devised methods that aim to automatically reduce and reason about
robots' resource footprints.
That work fits within the larger context of
methodologies and formalisms for tackling robot design problems, being useful
for designing robots subject to resource
limits\,\cite{censi17co,pervan2018low,saberifar18hardness}.  But, more
fundamentally, the associated algorithms also help identify the information
requirements of certain robot tasks.  The methods have the potential to provide
insights about the interplay of sensing, state, and actuation within the context
of particular tasks.  One class of objects where the problem of resource
minimization can be clearly posed is in the case of combinatorial
filters\,\cite{lavalle10sensing}.  These are discrete variants of the
probabilistic estimators and recursive Bayesian
filters widely adopted for practical use in robots.  Combinatorial filters
process a stream of discrete sensor inputs and integrate information via
transitions between states. The natural question, studied
in\,\cite{o2017concise}, then is: \emph{How few states are needed to realize
specified filter functionality?}  In this paper, we define a more general class
of filters and ask the same question.



We start with a simple motivating scenario where the generalization we
introduce is exactly what is needed.  Figure~\ref{fig:motivation_scenario}
shows a driving drone patrolling a house.\footnote{Such bizarre chimera robots
are not our invention, e.g., see the Syma X9 Flying Car.} The drone can either drive or
fly, but its choice must satisfy navigability constraints. Its wheels can't
drive on grass (\ssl{F} and \ssl{Y}) nor in the pantry (\ssl{P}), owing to spills.
Spinning propellers, on the other hand, will disturb the tranquil
bedroom~(\ssl{B}).  Otherwise, either means may be chosen (see inset map pair
marking regions in brown/blue for driving/flying).  The robot is equipped with
an ambient light sensor that is useful because the living room and kitchen are
lighter than the bedroom and pantry, while the outdoors is lightest of all.

\begin{figure}[t]
\centering
\includegraphics[scale=0.40]{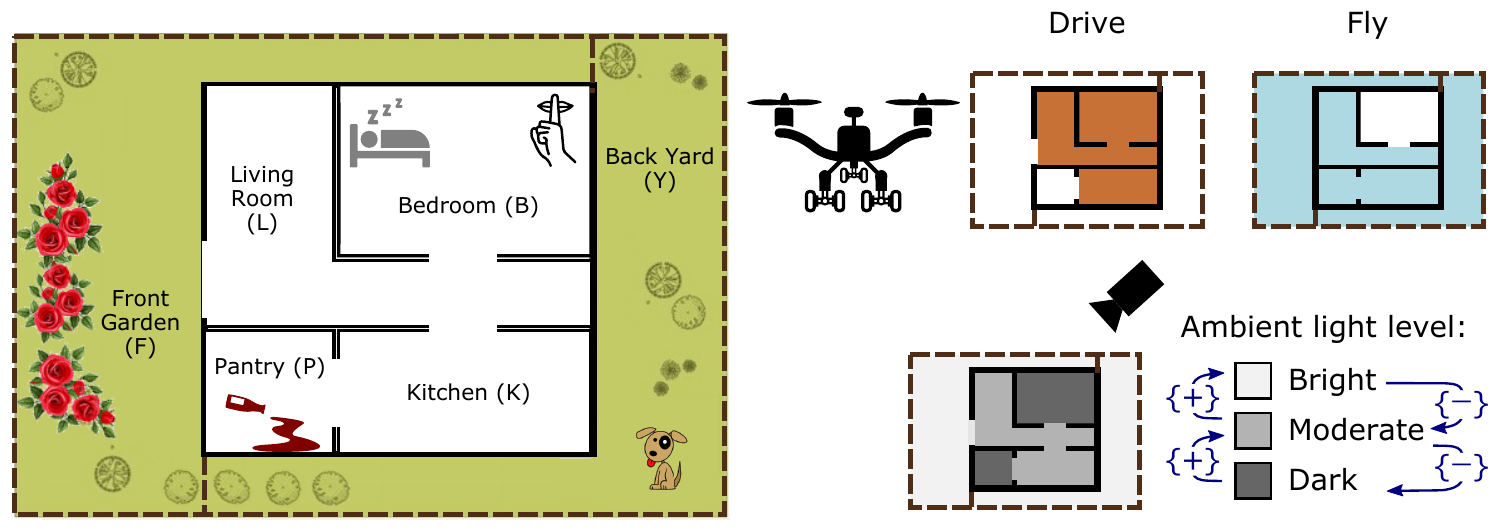}
\caption{A hybrid drone with a light sensor monitors a home. The robot is
capable of driving or flying: the grass outdoors (\ssl{F} and \ssl{Y}) and
liquids~(\ssl{P}) necessitate that it be airborne; the noise requirement means
it may only drive within the bedroom~(\ssl{B}).  It has a light sensor that
distinguishes three levels of ambient brightness.  Changes in brightness
(increasing `$+$', decreasing `$-$', same `$=$') provide the robot cues about
its location.  
\label{fig:motivation_scenario}}
\shortvspace{-20pt}
\end{figure}

\begin{figure}[ht]
\shortvspace{-20pt}
\begin{subfigure}[b]{0.30\textwidth}
\includegraphics[scale=0.5]{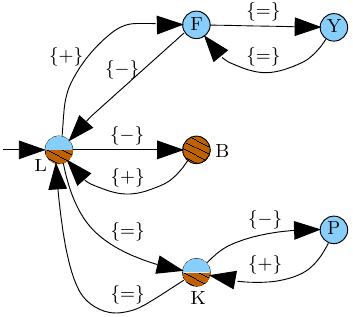}
\caption{A na\"\i ve filter, with one state per region, codifies all valid choices.\label{fig:motivation_input}}
\end{subfigure}
\hspace{0.5cm}
\begin{subfigure}[b]{0.29\textwidth}
\includegraphics[scale=0.5]{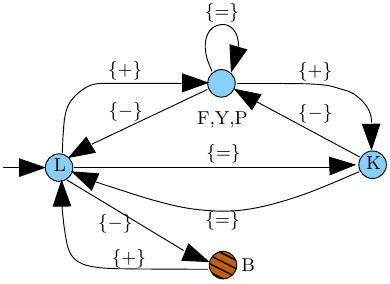}
\caption{A minimal filter when choosing to fly in the living room (\ssl{L}) and
kitchen~(\ssl{K}).\label{fig:motivation_myopic}}
\end{subfigure}
\hspace{0.5cm}
\begin{subfigure}[b]{0.30\textwidth}
\includegraphics[scale=0.5]{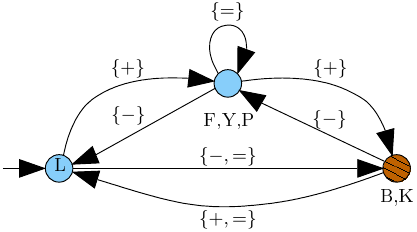}
\shortvspace{5mm}
\caption{A minimal filter when opting to fly in the living room (\ssl{L}) but not 
the kitchen~(\ssl{K}).\label{fig:motivation_min}}
\end{subfigure}
\caption{Combinatorial filters that tell the hybrid drone how to locomote.
The sequence of symbols {\footnotesize \,$\{$\ssl{$+$},\ssl{$-$},\ssl{$=$}$\}$\,} is traced on the graph,
and the color of the resultant state is the filter's output (blue for flight, brown for driving mode).
}
\shortvspace{-19pt}
\end{figure}

We wish to construct a filter for the drone to determine how to navigate, with
the inputs being brightness changes, and the filter's output providing some
valid mode of locomotion.  It is easy to give a valid filter by using one state
for each location\,---\,this na\"\i ve filter is depicted in
Figure~\ref{fig:motivation_input}.  In the living room and kitchen, the filter
lists two outputs since both modes are applicable there (both locations
are covered by the brown and blue choices).  Now consider the question of the smallest filter.
If we opt to fly in both the living room and kitchen, then the smallest filter
is shown in Figure~\ref{fig:motivation_myopic} with $4$ states.  But when
choosing to fly in the living room but drive in the kitchen, the minimal filter
requires only $3$ states (in Figure~\ref{fig:motivation_min}).

This last filter is also the globally minimal filter.  The crux is that states
with multiple valid outputs introduce a new degree of freedom which influences
the size of the minimal filter. These arise, for instance, whenever there are
`don't-care' options. The flexibility of such states must be retained to truly
minimize the number of states.

\shortvspace{-9pt}
\section{Preliminary Definitions and Problem Description}
\shortvspace{-8pt}

To begin, we define the filter minimization problem in the most general form,
where the input is allowed to be non-deterministic and each state may 
have multiple outputs. This is captured by the procrustean filter (p-filter)
formalism~\cite{setlabelrss}.

\shortvspace{-8pt}
\subsection{P-filters and their minimization}
\shortvspace{-4pt}

We firstly introduce the notion of p-filter:
\begin{definition}[procrustean filter~\cite{setlabelrss}]
A \defemp{procrustean filter}, \defemp{p-filter} or \defemp{filter} for short,
is a tuple $(V, V_0, Y, \tau, C, c)$ with:
\begin{tightenumerate}
\item a finite set of states $V$, a non-empty initial set of states
$V_0\subseteq V$, and a set of possible observations $Y$,
\item a transition function $\tau: V\times V\rightarrow 2^Y$,
\item a set $C$, which we call the output space, and 
\item an output function $c: V\to 2^{C}\setminus\{\emptyset\}$.
\end{tightenumerate}
\end{definition}
The states, initial states and observations for p-filter $F$ will be denoted 
$V(F)$, $V_0(F)$ and $Y(F)$. Without loss of generality, we will also treat a
p-filter as a graph with states as its vertices and transitions as directed
edges.

A sequence of observations can be traced on the p-filter:
\begin{definition}[reached]
Given any p-filter $F=(V, V_0, Y, \tau, C, c)$, a sequence of observations $s=y_1\dots y_n\in
Y^{*}$, and states $w_0, w_n\in V$, we say that $w_n$ is a state \defemp{reached by}
some sequence $s$ from $w_0$ in $F$ (or $s$ reaches $w_n$ from $w_0$), if there
exists a sequence of states $w_0, \dots, w_{n}$ in $F$, such that $\forall i\in
\{1, \dots, n\}, y_i\in \tau(w_{i-1}, w_i)$. We denote the set of all states
reached by $s$ from state $w_0$ in $F$ as $\reachedvf{F}{w_0}{s}$.
For simplicity, we use $\reachedv{F}{s}$, without the subscript, to denote the
set of all states reached when starting from any state in $V_0$, i.e.,
$\reachedv{F}{s}=\cup_{v_0\in V_0} \reachedvf{F}{v_0}{s}$.  Note that
$\reachedv{F}{s}=\emptyset$ holds only when sequence $s$ \emph{crashes} in $F$
starting from $V_0$.
\end{definition} 

For convenience, we will denote the set of sequences reaching state $v\in V$
from some initial state by $\reaching{F}{v}$.

\begin{definition}[extensions, executions and interaction language]
\label{def:ext}
An \defemp{extension} of a state $v$ on a p-filter $F$ is a finite sequence of
observations $s$ that does not crash when traced from $v$, i.e.,
$\reachedvf{F}{v}{s}\neq\emptyset$. An \defemp{extension} of any initial state 
$v_0\in V_0(F)$ is also called an \defemp{execution} or a \defemp{string} on
$F$. The set of all extensions of a state $v$ on $F$ is called the
\defemp{extensions} of $v$, written as $\extensions{F}{v}$. The
extensions of all initial vertices on $F$ is also called the \defemp{interaction
language} (or, briefly, just \defemp{language}) of $F$, and is written
$\Language{F}=\cup_{v_0\in V_0(F)} \extensions{F}{v_0}$.
\end{definition}

Note in particular that the empty string $\epsilon$ belongs to the extensions of
any state on the filter, and belongs to the language of the filter as well.

\begin{definition}[filter output]
Given any p-filter $F=(V, V_0, Y, \tau, C, c)$, a string $s$ and an output $o\in
C$, we say that $o$ is a \defemp{filter output} with input string $s$, if $o$ is an
output from the state reached by $s$, i.e., $o\in \cup_{v\in \reachedv{F}{s}} c(v)$.
We denote the set of all filter outputs for string $s$ as
$\reachedc{F}{s}=\cup_{v\in \reachedv{F}{s}} c(v)$.
\end{definition}
Specifically, for the empty string $\epsilon$, we have
$\reachedc{F}{\epsilon}=\cup_{v_0\in V_0(F)} c(v_0)$.

\begin{definition}[output simulating]
Given any p-filter $F$, a p-filter $F'$ \defemp{output simulates} $F$ if
$\forall s\in \Language{F}$, $\reachedc{F'}{s}\neq \emptyset$ and
$\reachedc{F'}{s}\subseteq \reachedc{F}{s}$.
\end{definition}
Plainly in words: for one p-filter to output simulate another, it has to generate some of
the outputs of the other, for every string the other admits.


We are interested in practicable p-filters with deterministic behavior: 
\begin{definition}[deterministic]
A p-filter $F=(V, V_0, Y, \tau, C, c)$ is \defemp{deterministic} or
\defemp{state-determined}, if $|V_0|=1$, and for every $v_1, v_2, v_3\in V$ with
$v_2\neq v_3$, $\tau(v_1, v_2)\cap \tau(v_1, v_3)=\emptyset$.
\end{definition}
Any non-deterministic p-filter $F$ can be state-determinized, denoted $\sde{F}$,
following Algorithm~$2$ in \cite{saberifar18pgraph}. 
 
Then the p-filter minimization problem can be formalized as follows:
\ourproblem{\textbf{P-filter Minimization (\pfm)}}
{A deterministic p-filter $F$.}
{A deterministic p-filter $F^{\dagger}$ with fewest states, such that
$F^{\dagger}$ output simulates $F$.
}
`{\sc pf}' denotes that the input is a p-filter, and `{\sc m}' denotes that we
are interested in finding a deterministic minimal filter as a solution.

\shortvspace{-12pt}
\subsection{Complexity of p-filter minimization problems}
\shortvspace{-3pt}

\begin{definition}[state single-outputting and multi-outputting] 
A p-filter $F=(V, V_0, Y, \tau, C, c)$ is \defemp{state single-outputting} or
\defemp{single-outputting} for short, if $c$ only maps to singletons, i.e., $|c(v)|=1, \forall v \in
V(F)$. Otherwise, we say that $F$ is \defemp{multi-outputting}. 
\end{definition}

Depending on whether the state in the input p-filter ({\sc pf}) is
single-outputting ({\sc so}) or multi-outputting ({\sc mo}), we further
categorize the problem \pfm into the following problems: \spfm and \mpfm.


\begin{lemma}
\label{lm:fm}
The filter minimization problem {\sc fm} of \cite{o2017concise} is \spfm, and
\spfm is NP-Complete (Theorem~$2$ in \cite{o2017concise}).  
\end{lemma}


\begin{theorem}
\mpfm is NP-Complete.
\end{theorem}
\begin{proof}
Firstly, \spfm is a special case of \mpfm problems. These \mpfm problems are at
least as hard as \spfm. Hence, \mpfm are in NP-hard. On the other hand, a solution
for \mpfm can be verified in polynomial time. (Change the equality check on line~7 of Algorithm~1 in~\cite{o2017concise} to a subset check.) Therefore, \mpfm is NP-Complete.
\end{proof}


Next, we examine related prior work on \spfm closely as a means
to develop new insights for our algorithms, first for \spfm (Section~\ref{sec:spfm}), and then \mpfm
(Section~\ref{sec:mpfm}).

\shortvspace{-12pt}
\section{Related work: Prior filter minimization ideas~(\spfm)}
\shortvspace{-8pt}

Several elements come together in this section and Figure~\ref{fig:method}
attempts to show the inter-relationships graphically.  The original question of
minimizing state in filtering is first alluded to by
LaValle~\cite{lavalle10sensing} as an open problem, who suggested that it is
`similar to Nerode equivalence classes'.  The problem of filter reduction,
i.e., \spfm in our terms, was formalized and shown to differ in complexity class
from the automata problem in \cite{o2017concise}. That paper also proposed a
heuristic algorithm, which served as a starting point for subsequent work.
The heuristic algorithm uses conflict graphs to designate which vertices cannot
be merged (are conflicting). It starts with a conflict relation where two
vertices are in conflict when they have different outputs, then iteratively
refines the conflict relation.  Refinement has two steps:
($i$)\,\emph{introducing edges}: two vertices are determined to be conflicting
or not via a graph coloring subroutine, and edges are added between conflicting
vertices; ($ii$)\,\emph{propagating conflicts upstream}: filter states are
marked as conflicted when they transition to conflicted states under the same
observation.  An example input filter, shown in Figure~\ref{fig:nonrecursive}, 
is reduced by following this procedure, which is depicted step-by-step in
Figures~\ref{fig:conflict_1}--\ref{fig:conflict_2_r}.

A conjecture in \cite{o2017concise} was that this algorithm is guaranteed to
find a minimal filter if the graph coloring subroutine gives a minimal
coloring.  (Put another way: the inexactness in arriving at a minimal filter
can be traced to the graph coloring giving a suboptimal result.) But this
conjecture was later proved to be false by Saberifar et
al.~\cite{saberifar2017combinatorial}.  They show an instance where there exist
multiple distinct optimal solutions to the graph coloring subproblem, only a
strict subset of which lead to the minimal filter. One might naturally ask, and
indeed they do ask, the question of whether some optimal coloring is sufficient
to arrive at the optimal filter.  Following along these lines 
(see \S7.3 in~\cite{saberifar2017combinatorial}),
one might sharpen the original conjecture of \cite{o2017concise} to give the
following statement:

\begin{figure}[ht!]
\centering
\includegraphics[scale=0.5]{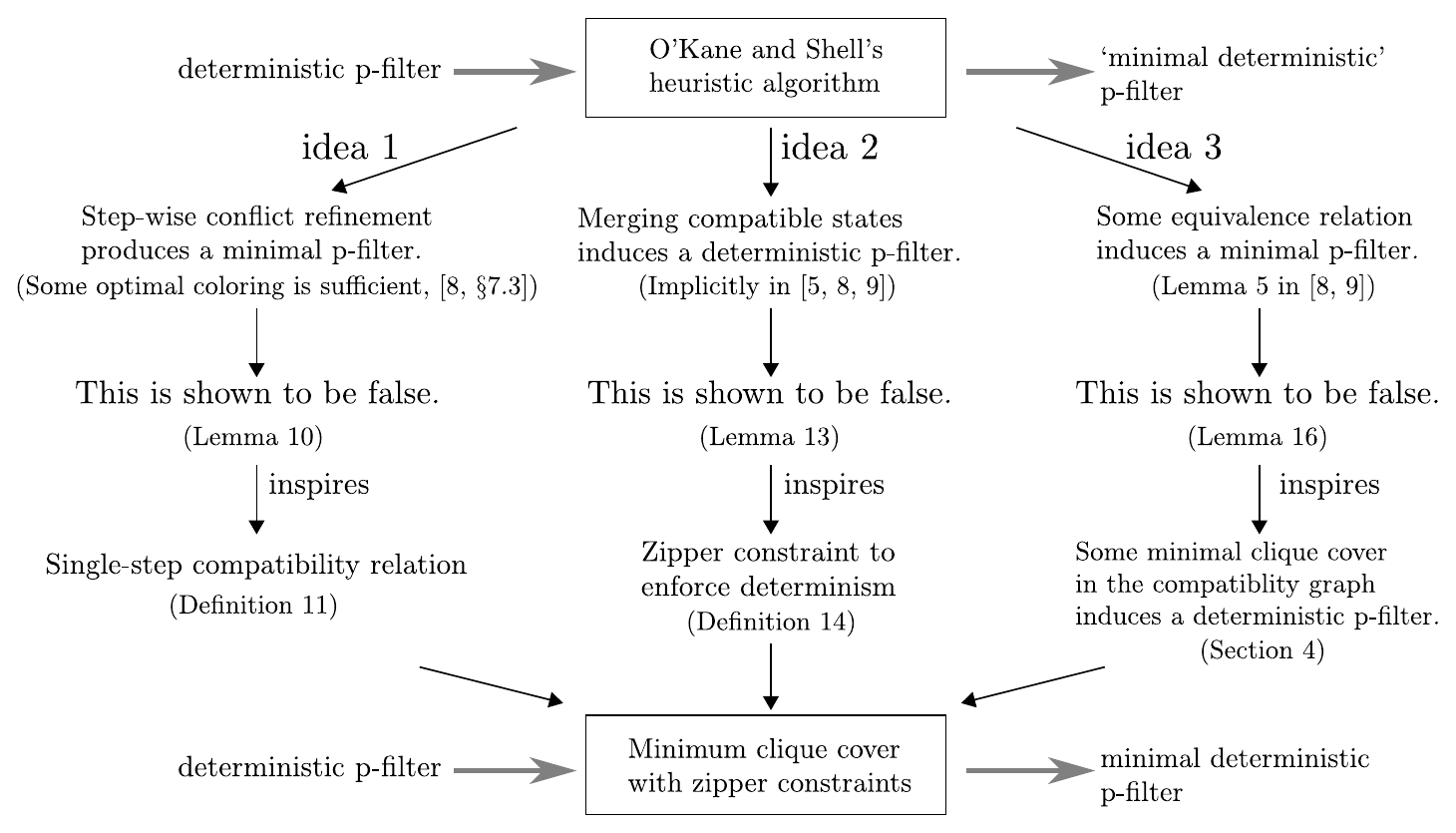}
\caption{
Three distinct insights led to the development of a new algorithm (described
in Section~\ref{sec:spfm}) for \spfm.  This roadmap shows the provenance of those insights
in terms of previous ideas in \spfm, which we examine carefully.
\label{fig:method}} 
\shortvspace{-12pt}
\end{figure}

\begin{idea} 
\label{idea:step-wise}
In the step-wise conflict refinement procedure of O'Kane and Shell's
heuristic algorithm~\cite{o2017concise}, some optimal coloring is
sufficient to guarantee a minimal filter for \spfm.
\end{idea}


\begin{lemma}
\label{lemma:counterexampleiterative}
Idea~\ref{idea:step-wise} is false.
\end{lemma}
\begin{proof}
This is simply shown with a counterexample. Consider the problem of minimizing the
input filter shown in Figure~\ref{fig:nonrecursive}, the heuristic algorithm
will first initialize the colors of the vertices with their output. Next, it
identifies the vertices that disagree on the outputs of extensions with length $1$
as shown in Figure~\ref{fig:conflict_1}, and then refines the colors of the
vertices as shown in Figure~\ref{fig:conflict_1_r} following a minimal graph
coloring solution on the conflict graph. Then it further identifies the
conflicts on extensions with length $2$, via the conflict graph shown in
Figure~\ref{fig:conflict_2}, and the vertex colors are further refined as
shown in Figure~\ref{fig:conflict_2_r}.  Now, no further conflicts can be
found. A filter, with $6$ states, is then obtained by merging the states with
the same color.  However, there exists a minimal filter, with $5$ states, shown
in Figure~\ref{fig:minimal_nonrecursive}, that can be found by choosing 
coloring solution for the conflict graph shown in Figure~\ref{fig:conflict_1}. 
That coloring is suboptimal.
\end{proof}

\begin{figure}[b!]
\begin{subfigure}[t]{0.55\textwidth}
\centering
\includegraphics[scale=0.5]{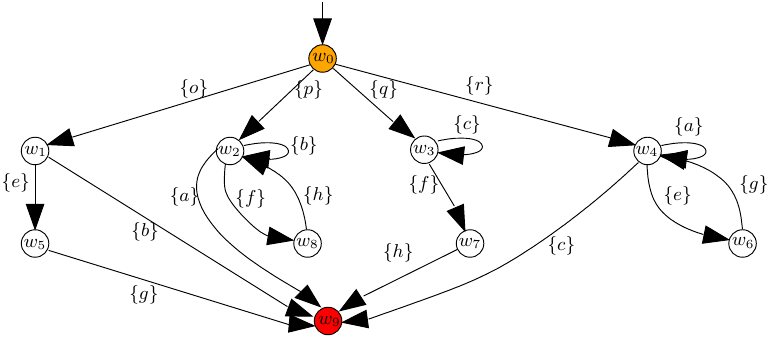}
\caption{Example input p-filter.\label{fig:nonrecursive}}
\end{subfigure}
\hfill
\begin{subfigure}[t]{0.4\textwidth}
\shortvspace{-1.8cm}
\raggedright
\includegraphics[scale=0.7]{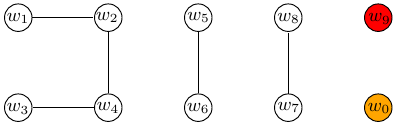}
\caption{Graphs of initial conflicts for all vertices with the same
output.\label{fig:conflict_1}}
\end{subfigure}
\\
\begin{subfigure}[t]{0.55\textwidth}
\centering
\includegraphics[scale=0.5]{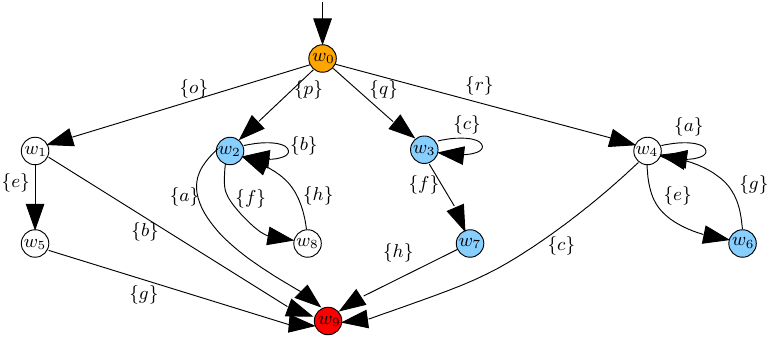}
\caption{The first refinement of the filter following an optimal coloring of the
conflict graphs.\label{fig:conflict_1_r}}
\end{subfigure}
\hfill
\begin{subfigure}[t]{0.4\textwidth}
\raggedright
\includegraphics[scale=0.7]{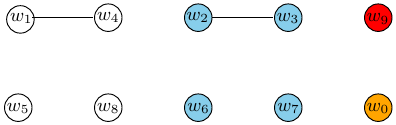}
\caption{Reduced conflict graphs.\label{fig:conflict_2}}
\end{subfigure}
\\
\begin{subfigure}[t]{0.48\textwidth}
\raggedleft
\includegraphics[scale=0.45]{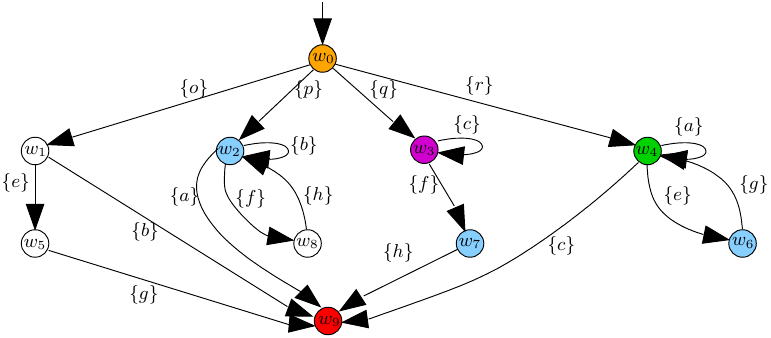}
\caption{Second refinement of filter following an optimal coloring of
the conflict graphs.\label{fig:conflict_2_r}}
\end{subfigure}
\hfill
\begin{subfigure}[t]{0.48\textwidth}
\raggedright
\includegraphics[scale=0.45]{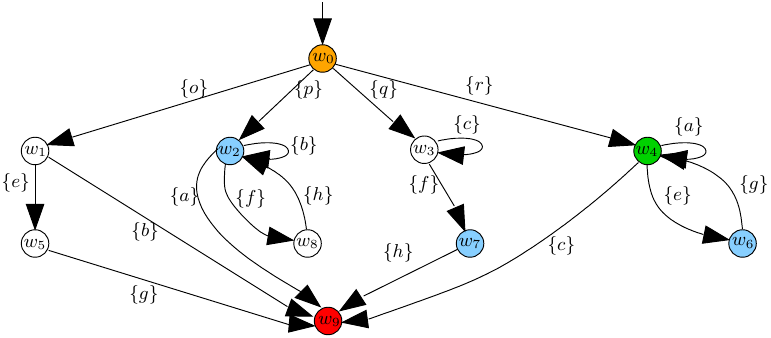}
\caption{The coloring that gives the minimal
filter.\label{fig:minimal_nonrecursive}}
\end{subfigure}
\caption{An example run of the heuristic minimization algorithm in
\cite{o2017concise} (a)--(e).
This particular input also 
shows that optimal step-wise conflict refinement may fail to yield a minimal filter (Lemma~\ref{lemma:counterexampleiterative}).}
\shortvspace{-12pt}
\end{figure}

This appears to indicate a sort of \emph{local optimum} arising via
sub-problems associated with incremental (or stepwise) reduction. Since optimal
colorings for individual steps are seen to be insufficient to guarantee a
minimal filter, to find a minimal filter, we would have to enumerate all
colorings (suboptimal or otherwise) at each iteration. That is, however,
essentially a brute force algorithm.  A more informed approach is to
compute implications of  conflicts more \emph{globally}, in a way that doesn't
depend on earlier merger decisions.  In our algorithm, rather than tracking
vertices which are in conflict, we introduce a new notion of compatibility
between vertices that may be merged.  This notion differs from the one
recursively defined in \cite{rahmani2018relationship}, as our compatibility
relation is computed in one fell swoop, before making any decisions to reduce
the filter: 

\begin{definition}[compatibility]
\label{def:compatibility}
Let $F$ be a deterministic p-filter. We say a pair of vertices $v,w\in V(F)$ are
\defemp{compatible}, denoted $v\compatible w$, if they agree on the outputs
of all their extensions, i.e., $\forall s\in
\extensions{F}{v}\cap\extensions{F}{w}, \forall v'\in \reachedvf{F}{v}{s},
\forall w'\in \reachedvf{F}{w}{s}, c(v')=c(w')$. 
A \emph{mutually compatible} set consists of vertices where all pairs
are compatible.
\end{definition}
Via this notion of compatibility, we get an undirected compatibility graph:
\begin{definition}[compatibility graph]
\label{def:compatibilitygraph}
Given a deterministic filter $F$, its compatibility graph
$\compatibilitygraph{F}$ is an unlabeled undirected graph constructed by
creating a vertex associated with each state in $F$, and building an edge
between the pair of vertices associated with two compatible states.
\end{definition}
This compatibility graph can be constructed in polynomial time. 
As every filter state and associated compatibility graph state are
one-to-one, to simplify notation we'll use the same symbol for both
and context to resolve any ambiguity.



The second idea relates to the type of the output one obtains after merging
states that are compatible or not in conflict.  Importantly, the filter
minimization problem \spfm requires one to give a minimal filter which
is deterministic.

\begin{idea}
\label{idea:merge}
By merging the states that are compatible, the heuristic algorithm always produces a deterministic p-filter.
\end{idea}

The definition of the reduction problems
within~\cite{o2017concise,saberifar2017combinatorial,rahmani2018relationship}
are specified so as to require that the output obtained be deterministic. 
But this postcondition is never shown or formally established.
In fact, it does not always hold.

\begin{lemma}
\label{lem:merge}
Idea~\ref{idea:merge} is false.
\end{lemma}
\begin{proof}
We show that the existing algorithm may produce a non-deterministic filter,
which does not output simulate the input filter, and is thus not a valid
solution. Consider the filter shown in Figure~\ref{fig:input_nd} as an input. The
vertices with the same color are compatible with each other, with the following exception for $w_5$, $w_6$
and $w_7$. Vertex $w_5$ is compatible with $w_6$, vertex $w_6$ is compatible with $w_7$, but
$w_5$ is not compatible with $w_7$. The minimal filter found by the existing
algorithm is shown in Figure~\ref{fig:false_nd}. 
The string $aac$ suffices to shows the non-determinism, reaching both
orange and cyan vertices. It fails to output simulate the input because
cyan should never be produced.
\end{proof}

If determinism can't be taken for granted, we might constrain the output to
ensure the result will be a deterministic filter.  To do this, we introduce a
\auxtext when merging compatible states:

\begin{definition}[\auxtext]\label{defn:aux_const}
In the compatibility graph $G=\compatibilitygraph{F}$ of filter $F$, if there
exists a set of mutually compatible states $U=\{u_1, u_2,\dots, u_n\}$, then
they can only be selected to be merged if they always transition to a set of
states that are also selected to be merged. For any sets of mutually compatible
states $U, W\subseteq V(G)$ and some observation $y$, we create a
\defemp{\auxtext} expressed as a pair $(U,W)_y$ if $W=\left\{w\in V(G)\mid
y\in\tau(u,w) \text{ for some } u \in U\right\}$. We denote the set of all
\auxestext on compatibility graph $G=\compatibilitygraph{F}$ by 
$\auxes{F}$.  
\end{definition}
The \auxestext for the input filter shown in Figure~\ref{fig:input_nd} consist
of $(\{w_1, w_2\}, \{w_5, w_6\})_a$ and $(\{w_3, w_4\}, \{w_6, w_7\})_b$.
Constraint $(\{w_1, w_2\}, \{w_5, w_6\})_a$ is interpreted as: if $w_1$ and
$w_2$ are selected for merger, then $w_5$ and $w_6$ (reached under $a$) should
also be merged. We call it a \auxtext owing to the resemblance to a zipper
fastener: merger of two earlier states, merges (i.e., pulls together) later
states. In the worst case, the number of \auxestext can be exponential in the
size of the input filter.

\begin{figure}[ht!]
\begin{subfigure}[b]{0.45\textwidth}
\includegraphics[scale=0.7]{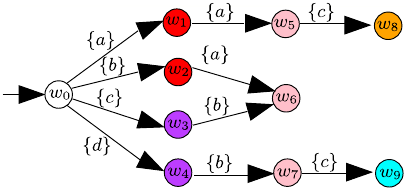}
\caption{An input filter.\label{fig:input_nd}}
\end{subfigure}
\hspace{0.2cm}
\begin{subfigure}[b]{0.45\textwidth}
\includegraphics[scale=0.7]{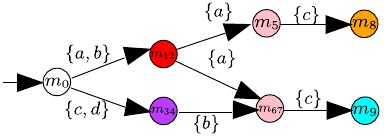}
\caption{The non-deterministic minimal filter found by the existing
algorithm.\label{fig:false_nd}} 
\end{subfigure}
\shortvspace{-8pt}
\caption{A counterexample showing how compatible merges may introduce
non-determinism (Lemma~\ref{lem:merge}). The input filter also illustrates a
violation of the presumption that an equivalence relation can yield a minimum
filter (Lemma~\ref{lem:eqr}).\label{fig:nondeterminism}} 
\shortvspace{-17pt}
\end{figure}

A third idea is used by O'Kane and Shell's heuristic algorithm and is also
stated, rather more explicitly, by Saberifar et al.~(see Lemma~5
in~\cite{saberifar2017combinatorial} and Lemma~$5$
in~\cite{rahmani2018relationship}). It indicates that we can obtain a minimal
filter via merging operations on the compatible states, which
yields a special class of filter minimization problems. 
For this class, recent work has expoited integer linear programming
techniques to compute exact
and feasible solutions efficiently~\cite{rahmani2020integer}.

\begin{idea}
\label{idea:equiv}
Some equivalence relation induces a minimal filter in \spfm.
\end{idea}

Before examining this, we rigorously define the notion of an induced relation:
\shortvspace{-14pt}
\begin{definition}[induced relation]
Given a filter $F$ and another filter $F'$, if $F'$ output simulates $F$,
then $F'$ induces a relation $R\subseteq V(F)\times V(F)$, where $(v, w)\in R$
if and only if there exists a vertex $v'\in V(F')$ such that
$\reaching{F}{v}\cap \reaching{F'}{v'}\neq\emptyset$ and $\reaching{F}{w}\cap
\reaching{F'}{v'}\neq\emptyset$. We also say that $v$ and $w$ corresponds to
state $v'$.
\end{definition}

\begin{lemma}\label{lem:eqr}
Idea~\ref{idea:equiv} is false.
\end{lemma}
\begin{proof}
It is enough to scrutinize the previous counterexample closely.
The minimization problem \spfm for the input filter shown in Figure~\ref{fig:input_nd}, is
shown in Figure~\ref{fig:minimal_nd}. It is obtained by ($i$) splitting vertex
$w_6$ into an upper part reached by $a$ and a lower part reached by $b$, ($ii$)
merging the upper part of $w_6$ with $w_5$, the lower part of $w_6$ with $w_7$, and
other vertices with those of the same color. This does not induce an equivalence relation,
since $w_6$ corresponds to two different vertices in the minimal filter. 
\end{proof}

\begin{figure}[ht!]
\shortvspace{-22pt}
\centering
\begin{subfigure}[b]{0.45\textwidth}
\includegraphics[scale=0.8]{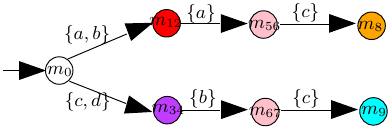}
\caption{A minimal filter for Figure~\ref{fig:input_nd}.\label{fig:minimal_nd}}
\end{subfigure}
\hfill
\begin{subfigure}[b]{0.45\textwidth}
\includegraphics[scale=0.7]{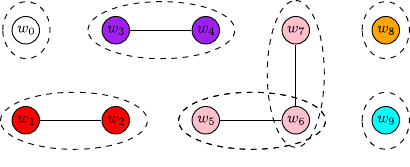}
\caption{Cliques from the minimal filter.\label{fig:compatibility_split}} 
\end{subfigure}
\caption{A minimal filter for Figure~\ref{fig:input_nd} and its induced
cliques.}
\shortvspace{-22pt}
\end{figure}

In light of this, for some filter minimization problems, there may be no quotient operation
that produces a minimal filter and an exact algorithm for minimizing filters
requires that we look beyond equivalence relations.  

Some strings that reach a single state in an input filter may reach multiple
states in a minimal p-filter (e.g., $ba$ and $cb$ on
Figure~\ref{fig:input_nd} and~\ref{fig:minimal_nd}).  On the other hand,
strings that reach different states in the input p-filter may reach the same
state in the minimal filter (e.g., $a$ and $b$ on those same filters).  We say
that a state from the input filter corresponds to a state in the minimal filter
if there exists some string reaching both of them and, hence, this
correspondence is many-to-many.
An important observation is this: for each state $s$ in
some hypothetical minimal filter, suppose we collect all those states in the
input filter that correspond with~$s$.  When we examine the associated states in
the compatibility graph for that collection, they must all form a clique.  Were
it not so, the minimal filter could have more than one output associated for
some strings owing to non-determinism.  But this causes it to fail to output
simulate the input p-filter. 

After firming up and developing these intuitions, the next section 
introduces the concept of a clique cover which enables representation of a search
space that includes relations more general than equivalence
relations.  Based on this new representation, we propose a graph problem
use of \auxestext, and prove it to be equivalent to 
filter minimization.


\shortvspace{-12pt}
\section{A new graph problem that is equivalent to \spfm}
\label{sec:spfm}
\shortvspace{-6pt}

By building the correspondence between the input p-filter in
Figure~\ref{fig:input_nd} and the minimal result in
Figure~\ref{fig:minimal_nd}, one obtains the set of cliques in the
compatibility graph shown visually in Figure~\ref{fig:compatibility_split}.
Like previous approaches that make state merges by analyzing the compatibility
graph, we interpret each clique as a set of states to be merged into one state
in the minimal filter.  The clique containing $w_3$ and $w_4$ in
Figure~\ref{fig:compatibility_split} gives rise to $m_{34}$ in the minimal
filter in Figure~\ref{fig:minimal_nd} (and $w_1$ and $w_2$ yields $m_{12}$, and
so on).  
However, states may further be shared across multiple cliques. We observe that
$w_6$ was merged with $w_5$ in the minimal filter to give $m_{56}$, and $w_6$
also merged with $w_7$ to give $m_{67}$. The former has an incoming edge
labeled with an~$a$, while the latter has an incoming edge labeled~$b$.  The
vertex $w_6$, being shared by multiple cliques, is split into different copies
and each copy merged separately.

Generalizing this observation, we turn to searching for the smallest set of
cliques that cover all vertices in the compatibility graph. Further, to
guarantee that the set of cliques induces a deterministic filter, we must
ensure they respect the \auxestext.  It will turn out that a
solution of this new constrained minimum clique cover problem always induces a
minimal filter for \spfm, and a minimal filter for \spfm always
induces a solution for this new problem.  The final step is to reduce any \mcca
problem to a SAT instance, and leverage SAT solvers to find a minimal filter
for \spfm.

\shortvspace{-8pt}
\subsection{A new minimum clique cover problem}
\shortvspace{-4pt}

To begin, we extend the preceding argument from the compatibility clique
associated to single state $s$, over to all the states in the minimal filter.
This leads one to observe that the collection of all cliques for each state
in the minimal p-filter forms a clique cover:
\begin{definition}[induced clique cover]
Given a p-filter $F$ and another p-filter $F'$, we say that a vertex $v$ in $F$
\defemp{corresponds to} a vertex $v'_i$ in $F'$ if \mbox{$\reaching{F}{v}\cap
\reaching{F'}{v'_i}\neq\emptyset$}. 
Then, denoting the subset of vertices of $F$ corresponding to $v'_i$ in $F'$
with $K_{v'_i}=\{v\in V(F)\,|\, v \text{ corresponds to } v'_i\}$, we form the
collection of all such sets, $\cover{Q}(F, F')=\{K_{v'_1}, K_{v'_2}, \dots, K_{v'_n}\}$,
for $i\in\{1,\dots, n\}$ where $n=|V(F')|$. 
When $F'$ output simulates $F$, then the $K_{v'_i}$ form cliques in the compatibility
graph~$\compatibilitygraph{F}$. Further, when this collection of sets $\cover{Q}(F,
F')$ covers all vertices in $F$, i.e., $\cup_{K_i\in \cover{Q}(F, F')}=V(F)$, we say
that $\cover{Q}(F, F')$ is an induced clique cover. 

\end{definition}

It is worth repeating: the size of filter $F'$ (in terms of number of
vertices) and the size of the induced clique cover (number of sets) are equal.

Without loss of generality, here and henceforth we only consider the p-filter
with all vertices reachable from the initial state, since the ones that can
never be reached will be deleted during filter minimization anyway.

Each clique of the clique cover represents the states that can be potentially
merged. But the \auxtext, to enforce determinism, requires that the
set of vertices to be merged should always transition under the same observation
to the ones that can also be merged. Hence, the \auxestext (of
Definition~\ref{defn:aux_const}) can be evaluated across whole covers:
\begin{definition}
A clique cover $\cover{K}=\{K_0, K_1, \dots, K_m\}$ satisfies the set of
\auxestext $\Aux=\{(U_1, W_1)_{y_1}, (U_2, W_2)_{y_2},\dots\}$, when
for every \auxtext $(U_i, W_i)_{y_i}$, if there exist a clique
$K_s\in \cover{K}$, such that $U_i\subseteq K_s$, then there exists another
clique $K_t\in \cover{K}$ such that $W_i\subseteq K_t$.  
\end{definition}


Now, we have our new graph problem, \mcca.
\ourproblem{Minimum clique cover with \auxestext (\mcca)}
{A compatibility graph $G$, a set of \auxestext \Aux.}
{A minimal cardinality clique cover of $G$ satisfying \Aux.}

\shortvspace{-8pt}
\subsection{From minimal clique covers to filters}
\shortvspace{-4pt}
Given a minimal cover that solves \mcca, we construct a
filter by merging the states in the same clique and choosing 
edges between these cliques appropriately:
\shortvspace{-12pt}
\begin{definition}[induced filter]
Given a clique cover $\cover{K}$ on the compatibility graph of deterministic p-filter
$F$, if $\cover{K}$ satisfies all the \auxestext in $\auxes{F}$, then it
\defemp{induces a filter} $F'=M(F, \cover{K})$ by treating cliques as vertices:
\begin{enumerate}
\item Create a new filter $F'=(V', V'_0, Y, \tau', C, c')$ with $|\cover{K}|$
vertices, where each vertex $v'$ is associated with a clique $K_{v'}$ in $\cover{K}$;
\item Add each vertex $v'$ in $F'$ to $V'_0$ iff the associated clique
contains an initial state in $F$;
\item The output of every $v'$ in $F'$, with associated clique $K_{v'}$,
is the set of common outputs for all states in $K_{v'}$, i.e., $c'(v')=\cap_{v\in K_{v'}} c(v)$. 
\item For any pair of $v'$ and $w'$ in $F'$, inherit all transitions
between states in the cliques of $v'$ and $w'$, i.e., $\tau(v', w')=\cup_{v\in
K_{v'}, w\in K_{w'}} \tau(v,w)$.  
\item For each vertex $v'$ in $F'$ with multiple outgoing edges labeled
$y$, keep only the single edge to the vertex $w'$, such that all vertices $K_{v'}$ transition
to under $y$ are included in $K_{w'}$.
This edge must exist since $\cover{K}$ satisfies all $\auxes{F}$.\label{step:last}
\end{enumerate}
\end{definition}

The size of the cover (in terms of number of sets) and size of the induced
filter (number of vertices) are equal.
 
Notice that the earlier intuition is mirrored by this formal construction:
states belonging to the same clique are 
merged when constructing the induced filter; states in multiple
cliques are split when we make the edge choice in step~\ref{step:last}.
Next, we establish that the induced filter indeed supplies the goods: 

\begin{restatable}{lemma}{deterministicouputsimulating}
Given any clique cover $\cover{K}$ on the compatibility graph $\compatibilitygraph{F}$ of a deterministic
p-filter $F$, if $\cover{K}$ satisfies the \auxestext $\auxes{F}$ and covers
all vertices of $\compatibilitygraph{F}$, then the induced filter
$F'=M(F,\cover{K})$ is deterministic and output simulates $F$.
\label{lm:clique_aux_deterministic_output_simu}
\end{restatable}
\ifdefined\arxiv
	\begin{proof}
	For any string $s\in \Language{F}$, let the vertex reached by string $s$ in $F$
	be $v$. Then $v$ must belong at least one clique in $\cover{K}$, where all vertices in
	this clique can be viewed as merged into a new vertex in $F'$. Hence, $s$ should
	reach at least one vertex in $F'$ and this vertex yield the same output
	$\reachedc{F}{s}$.
	Since $\cover{K}$ satisfies the \auxestext $\auxes{F}$, the induced filter
	$F'$ must be deterministic since there is no vertex that has any
	non-deterministic outgoing edges bearing the same label. 
	Because $F'$ is deterministic, $\forall s\in\Language{F}$, $s$ reaches a single
	vertex in $F'$. In addition, this vertex in $F'$ shares the same output
	$\reachedc{F}{s}$. Therefore, $F'$ also output simulates $F$.
	\end{proof}
\else
\shortvspace{-4pt}
	\proofinextension
\shortvspace{4pt}
\fi

A surprising aspect of the preceding is how the \auxestext\,---which are imposed to ensure that a deterministic filter is
produced---\,enforce output-simulating behavior, albeit indirectly, too. One might have
expected that this separate property would demand a second type of constraint,
but this is not so.


On the other hand, needing to satisfy the \auxestext of the
input filter does not entail the imposition of any gratuitous requirements:

\shortvspace{-4pt}
\begin{restatable}{lemma}{cliqueaux}
Given any deterministic p-filters $F$ and $F'$, if $F'$ output simulates $F$,
then the induced clique cover $\cover{Q}(F,F')$ on the compatibility graph of $F$
satisfies all \auxestext in $\auxes{F}$.
\label{lm:clique_aux}
\end{restatable}
\ifdefined\arxiv
	\begin{proof}
	Suppose that $\cover{Q}(F, F')$ does not satisfy all \auxestext in
	$\auxes{F}$. Specifically, let $(U, V)_y\in \auxes{F}$ be the \auxtext that is
	violated, where each vertex in $V$ transitions from some vertex in $U$ under
	observation $y$. Then there exists a clique $K_s\in
	\cover{Q}(F, F')$, such that $U\subseteq K_s$, but there is no clique
	$K_j\in \cover{Q}(F,F')$ that $V\subseteq K_j$.  According to the construction
	of the induced cover, there exists a vertex $v'_s\in F'$, such that $K_s$
	corresponds to $v'_s$.  For any vertex $u_1\in K_s$, let $s_1\in
	\reaching{F}{u_1}\cap \reaching{F'}{v'_s}$. Then $s_1y$ is also a string in both
	$F$ and $F'$ since $u_1$ transitions to some vertex $v_1$ in $V$ under
	observation $y$ in $F$ and $F'$ is output simulating $F'$. Let
	$\reachedv{F'}{s_1y}=\{v'_t\}$. (It is a singleton set as $F'$ is deterministic.) Hence $v_1$ corresponds to
	$v'_t$ on common string $s_1y$. Similarly, each vertex $v\in V$ corresponds to
	$v'_t$ on some string ending with $y$. Let the clique corresponding to $v'_t$
	be $K_t$, and we have $K_t\supseteq V$. But that is a contradiction.
	\end{proof}
\else
\shortvspace{-4pt}
	\proofinextension
\fi

\shortvspace{-12pt}
\subsection{Correspondence of \mcca and \spfm solutions}
\shortvspace{-6pt}
To establish the equivalence between \mcca and \spfm, we will show that the
induced filter from the solution of \mcca is a minimal filter for \spfm, and the
induced clique cover from a minimal filter is a solution for \mcca.
%
%

\begin{lemma}
Minimal clique covers for \mcca induce minimal filters for~\spfm.
\label{lm:cover_to_filter}
\end{lemma}
\ifdefined\arxiv
	\begin{proof}
	Given any minimal clique cover $\cover{K}=\{K_1, K_2, \dots, K_m\}$ as a
	solution for problem \mcca with input p-filter $F$, construct p-filter $F'=M(F,
	\cover{K})$.  Since $\cover{K}$ satisfies the \auxestext $\auxes{F}$, $F'$ is
	deterministic and output simulates $F$ according to
	Lemma~\ref{lm:clique_aux_deterministic_output_simu}.  To show that $F'$ is a
	minimal deterministic filter for \spfm, suppose the contrary.  Then there exists
	a minimal deterministic filter $F^{\star}$ with fewer states, i.e.,
	$|V(F^{\star})|<|V(F')|$. Hence, $F^{\star}$ induces a clique cover
	$\cover{K^{\star}}$ with fewer cliques than $\cover{K}$. Since $F^{\star}$ is
	deterministic, $\cover{K^{\star}}$ satisfies all $\auxes{F}$ via
	Lemma~\ref{lm:clique_aux}.  But then $\cover{K^{\star}}$ satisfies all the
	requirements to be a solution for \mcca, and has fewer cliques than $\cover{K}$,
	contradicting the assumption.
	\end{proof}
\else
\shortvspace{-8pt}
	\proofinextension
\fi
\begin{lemma}
A minimal filter for \spfm with input $F$ induces a clique cover that
solves \mcca with compatibility graph and \auxestext of~$F$.  
\label{lm:filter_to_cover}
\end{lemma}
\ifdefined\arxiv
	\begin{proof}
	Given minimal filter $F^{\star}$ as a solution for \spfm with input filter $F$,
	we can construct a clique cover $\cover{K}=\cover{Q}(F, F^{\star})$ from the
	minimal filter. For this cover to be a solution for \mcca with compatibility
	graph $G=\compatibilitygraph{F}$ and \auxestext $\auxes{F}$, first, it must
	satisfy all constraints in $\auxes{F}$.  Lemma~\ref{lm:clique_aux} affirms this
	fact. Second, we must show it to be minimal among all the covers satisfying
	those constraints.  Supposing $\cover{K}$ is not a minimal, there must exist a
	clique cover $\cover{K'}$ with $|\cover{K'}|<|\cover{K}|$ satisfying
	$\auxes{F}$.  Then, consider the induced filter $F'=M(F, \cover{K'})$. Since
	$\cover{K'}$ satisfies all the \auxestext $\auxes{F}$, $F'$ is deterministic and
	will output simulate $F$ (Lemma~\ref{lm:clique_aux_deterministic_output_simu}).
	But $|V(F')| = |\cover{K'}| <  |\cover{K}| = |V(F^{\star})|$, contravening the
	fact that $F^{\star}$ is a minimal filter. Hence $\cover{K}$ is minimal.
	\end{proof}
\else
\shortvspace{-8pt}
	\proofinextension
\fi
\shortvspace{6pt}

Together, they establish the theorem.
\begin{theorem}
The solution for \mcca with compatibility graph and \auxestext of a
filter $F$ induces a solution for \spfm with input filter $F$, and vice versa.
\end{theorem}
\begin{proof}
Lemma~\ref{lm:cover_to_filter} and
Lemma~\ref{lm:filter_to_cover} comprise the complete result.
\end{proof}
Having established this correspondence, any 
\spfm can be solved by tackling its associated \mcca problem, the 
latter problem being cast as a SAT instance and solved via a solver
(see Section~\ref{section:reduction} for further details). 
%
%
%

\shortvspace{-11pt}
\section{Generalizing to \mpfm}
\shortvspace{-8pt}

\label{sec:mpfm}
Finally, we generalize the previous algorithm to multi-outputting filters.
In \mpfm problems, the input p-filter is deterministic but
states in the p-filter may have multiple outputs. 
One straightforward if unsophisticated
approach is to enumerate all filters under different output
choices for the states with multiple outputs, and then solve every one the
resulting deterministic single-outputting filters as instances of \spfm. The
filter with the fewest states among all the minimizers could then be treated as
a minimal one for the \mpfm problem.
\begin{figure}[ht!]
\begin{subfigure}[b]{0.55\textwidth}
\centering
\includegraphics[scale=0.7]{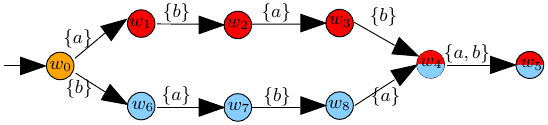}
\caption{An input filter with multi-outputting states. The colors in each
vertex represents its output.\label{fig:split_choose_input}}
\end{subfigure}
\hfill
\begin{subfigure}[b]{0.43\textwidth}
\centering
\includegraphics[scale=0.7]{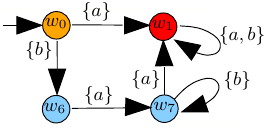}
\caption{A minimal filter when choosing to output the same color for $w_4$ and
$w_5$.\label{fig:split_same}}
\end{subfigure}
\begin{subfigure}[b]{0.55\textwidth}
\centering
\includegraphics[scale=0.7]{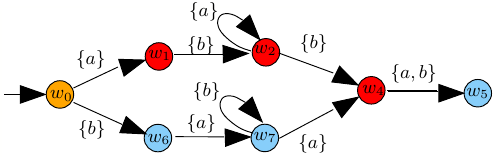}
\caption{A minimal filter when choosing to output different colors for $w_4$ and
$w_5$.\label{fig:split_diff}}
\end{subfigure}
\hfill
\begin{subfigure}[b]{0.43\textwidth}
\centering
\includegraphics[scale=0.7]{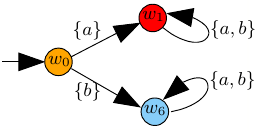}
\caption{A minimal filter for the input filter.\label{fig:split_min}}
\end{subfigure}
\caption{A multi-outputting filter minimization problem.}
\shortvspace{-16pt}
\end{figure}

Unfortunately, this is too simplistic.  Prematurely committing to an output
choice is detrimental.  Consider the input filter shown in
Figure~\ref{fig:split_choose_input}, it has two multi-outputting states ($w_4$
and $w_5$). If we choose to have both $w_4$ and $w_5$ give the same output, the
\spfm minimal filter, shown in Figure~\ref{fig:split_same}, has $4$ states. If
we choose distinct outputs for $w_4$ and $w_5$, the \spfm minimal filter, shown
in Figure~\ref{fig:split_diff}, now has $7$ states. But neither is the minimal
\mpfm filter. The true minimizer appears in Figure~\ref{fig:split_min}, with
only $3$ states. It is obtained by splitting both $w_4$ and $w_5$ into two
copies, each copy giving a different output.

The idea underlying a correct approach is that output choices should be made
together with the splitting and merging operations during filter minimization.
Multi-outputting vertices may introduce additional split operations, but
these split operations can still be treated via clique covers on the
compatibility graph. This requires that we define a new compatibility
relationship---it is only slightly more general than
Definition~\ref{def:compatibility}: 
\begin{definition}[group compatibility]
Let $F$ be a deterministic p-filter.  We say that the set of states $U=\{u_1,
u_2, \dots, u_n\}$ are \defemp{group compatible}, if there is a common output on
all their extensions, i.e., 
\shortvspace{-6pt}
\[\shortenXor{}{\small}
\forall s\in \bigcup\limits_{u \in U}
\extensions{F}{u}, 
\bigcap\limits_{w' \in W'}\!c(w')\neq \emptyset, 
\text{ where } W' = \reachedvf{F}{u_1}{s} \cup \reachedvf{F}{u_2}{s} \cup \cdots \reachedvf{F}{u_n}{s}.
\]
\end{definition}
\ifdefined\arxiv
(The preceding exploits the subtle fact, in Definition~\ref{def:ext}, that 
$\reachedvf{F}{v}{s} = \emptyset$ when tracing $s$ from $v$ crashes in $F$.)
With this definition, the compatibility graph must be generalized suitably:
\begin{definition}[compatibility simplicial complex]
Given a deterministic multi-output filter $F$, its \defemp{compatibility simplicial
complex} is a collection of simplices, where each simplex is a set of group compatible vertices in
$F$.
\end{definition}
The \auxestext are generalized too, replacing mutual compatibility
with group compatible states:
\begin{definition}[generalized \auxtext]\label{defn:g_aux_const}
In the compatibility simplicial complex of filter $F$, if there
exists a set of group compatible states $U=\{u_1, u_2,\dots, u_n\}$,
then they can only be selected to be merged if they always transition to a set
of states that are also selected to be merged. For any sets of group 
compatible states $U, W\subseteq V(G)$ and some observation $y$, we create a
\defemp{generalized \auxtext} expressed as a pair $(U,W)_y$ if $W=\{w\in
V(G)\mid y\in\tau(u,w)$ for some $u \in U\}$.  
\end{definition}

The information formerly encoded in cliques of edges is now within simplicies;
the minimum clique cover on the compatibility graph, thus, becomes a minimum
simplex cover on the compatibility simplicial complex.
Hence, the \mcca problem is generalized as follows:

\ourproblem{Generalized Minimum Cover with \auxestext (\gmca)}
{A compatibility simplicial complex $M$, a set of generalized \auxestext \Aux.}
{A minimal cardinality simplex cover of $M$ satisfying \Aux.}
\else
	With this definition, the minimization of a deterministic multi-outputting
	filter can also be written and solved as a slightly generalized \mcca problem,
	where ($i$) the compatibility graph is generalized toward a \defemp{compatibility
	simplicial complex}; ($ii$) to capture the set of vertices that can be merged,
	the clique in the \mcca problem is generalized to be a simplex in the
	\defemp{compatibility simplicial complex}; ($iii$) the \auxestext are redefined
	by replacing ``mutually compatible'' with the ``group compatible''
	states; ($iv$) the objective is to find the set of simplices that covers all
	vertices in the compatibility simplicial complex.  This generalized \mcca
	problem, termed \gmca, will be solved in the next section.

\fi
%



\shortvspace{-12pt}
\section{Reduction from \mcca and \gmca to SAT}
\label{section:reduction}
\shortvspace{-4pt}

Prior algorithms for filter minimization used multiple stages to find a set of
vertices to merge, solving a graph coloring problem repeatedly as more
constraints are identified.  In contrast, an interesting aspect of \mcca and
its generalized version is that it tackles filter minimization as a constrained
optimization problem with all constraints established upfront.  Thus the clique
(simplex) perspective gives an optimization problem which is tangible and easy
to visualize. Still, being a new invention, there are no solvers readily
available for direct use. But reducing \mcca (\gmca) to Boolean satisfaction (SAT)
enables the use of state-of-the-art solvers to find minimum cliques
(simplices).  

\ifdefined\arxiv
	Since the edges that make up the cliques can be encoded via 1-simplices,
in what follows we give the treatment of \gmca.

Next, we follow the standard practice for treating optimization problems via a
decision problem oracle, \emph{viz.} define a $k$-\gmca problem, asking for the
existence of a cover with size $k$ satisfying the \auxestext;
one then decreases $k$ to find the minimum cover. Each $k$-\gmca problem can be
written as a logic formula in conjunctive normal form (CNF), polynomial in the
size of the $k$-\gmca instance, and solved.  

Firstly, the $k$-\gmca problem is formalized as follows:

\ourproblem{Minimum simplex cover with \auxestext ($k$-\gmca)}
{A compatibility simplicial complex $M$, a set of \auxestext \Aux, maximum number
of simplices $k$} 
{A simplex cover with no more than $k$ simplices on $M$ that
satisfies all \auxestext in \Aux}

Then, we represent the simplex cover as choices to assign each vertex
$v$ in the compatibility simplicial complex to a simplex $i$, with $1\leq i\leq
k$.  To represent these choices, we create a boolean variable $R_v^i$ to
represent the fact that $v$ is assigned to simplex $i$, and its
negation $\negation{R_v^i}$ to represent its inverse. The simplex
cover is
captured by $k\times |V(M)|$ such variables.

A simplex cover for problem $k$-\gmca should guarantee that each vertex
in $M$ is assigned to at least one clique, i.e.,
\begin{equation}
\label{eq:one}
\prod_{v\in V(M)} \sum_{1\leq i\leq k} R_v^i.
\end{equation}
For simplicity, we will use ``$+$" and ``$\sum$" for logic or, ``$\cdot$" and
``$\prod$" for logic and, ``$=$" for logic equivalence.

We denote each simplex in the compatibility simplicial
complex as a set. Let the set of simplicies for compatibility 
simplicial complex $M$ be $E(M)$, which is a collection of sets.  For any set
$\{v_1, v_2, \dots, v_n\}$ that does not form a simplex in $M$, i.e.,
$\bar{S}=2^{V(M)}\setminus (E(M)\cup \{\emptyset\})$, they should never be assigned to
the same simplex:
\begin{equation}
\label{eq:two}
\prod_{\{v_1, v_2, \dots, v_n\}\in \bar{S}} \prod_{1\leq i\leq k}
(\negation{R_{v_1}^i}+\negation{R_{v_2}^i}+\dots+\negation{R_{v_n}^i}).
\end{equation}


To satisfy each \auxtext $(U, W)_y\in \Aux$ (with $|U|=m$ and $|W|=n$), if $U$
is assigned to a simplex $i$ ($1\leq i\leq k$), then there must exist
another simplex $j$ ($1\leq j\leq k$), such that all vertices in $W$ are
assigned to simplex $j$,
i.e.,
\begin{equation*}
\begin{aligned}
&\sum_{1\leq i\leq k} \prod_{u\in U}R_u^i\implies \sum_{1\leq j\leq k}
\prod_{w\in W}R_{w}^j&\\
=&\prod_{1\leq i\leq k} \sum_{u\in U}\negation{R_u^i}+\sum_{1\leq
j\leq k} \prod_{w\in W}R_{w}^j.
\end{aligned}
\end{equation*}
Let $X_i=\sum_{u\in U}\negation{R_u^i}$ and $Y=\sum_{1\leq
j\leq k} \prod_{w\in W}R_{w}^j$, then
\begin{equation*}
\begin{aligned}
\prod_{1\leq i\leq k} \sum_{u\in U}\negation{R_u^i}+\sum_{1\leq
j\leq k} \prod_{w\in W}R_{w}^j  = & X_1\cdot X_2\cdot\, \cdots\, \cdot X_k+ Y\\
 = & (X_1+ Y)\cdot (X_2+ Y)\cdot \,\cdots\, \cdot (X_k+ Y)  \\
 = & \prod_{1\leq i\leq k} (X_i+ Y).
\end{aligned}
\end{equation*}
For each $X_i+ Y$, we have 
\begin{equation*}
\begin{aligned}
Y+X_i&=\sum_{1\leq j\leq k} \prod_{w\in W}R_{w}^j+\sum_{u\in U}\negation{R_u^i}\\
&=\underbrace{R_{w_1}^1\cdot R_{w_2}^1\cdot \,\cdots\, \cdot
R_{w_n}^1}_{z^i_1}+\underbrace{R_{w_1}^2\cdot R_{w_2}^2\cdot \,\cdots\, \cdot
R_{w_n}^2}_{z^i_2}+\\
&\dots+\underbrace{R_{w_1}^k\cdot R_{w_2}^k\cdots
R_{w_n}^k}_{z^i_k}+\underbrace{\negation{R_{u_1}^i}\cdot
\text{True}}_{z^i_{k+1}}+\underbrace{\negation{R_{u_2}^i}\cdot
\text{True}}_{z^i_{k+2}}
+\dots+\underbrace{\negation{R_{u_m}^i}\cdot \text{True}}_{z^i_{k+m}}\\
&=(\sum_{1\leq \ell\leq k+m} z^i_{\ell})\cdot (\prod_{1\leq j\leq k} \prod_{1\leq
p\leq n} (\negation{z^i_{j}}+R_{w_p}^j))\cdot (\prod_{1\leq q\leq m}
(\negation{z^i_{k+q}}+\negation{R_{u_q}^i})).
\end{aligned}
\end{equation*}

Therefore, 
\begin{equation}
\label{eq:three}
\begin{aligned}
&\prod_{1\leq i\leq k} \sum_{u\in U}\negation{R_u^i}+\sum_{1\leq
j\leq k} \prod_{w\in W}R_{w}^j \\
=&\prod_{1\leq i\leq k} (X_i+Y)\\
=&\prod_{1\leq i\leq k}\Big[(\sum_{1\leq \ell\leq k+m} z^i_{\ell})\cdot (\prod_{1\leq j\leq k} \prod_{1\leq
p\leq n} (\negation{z^i_{j}}+R_{w_p}^j))\cdot (\prod_{1\leq q\leq m}
(\negation{z^i_{k+q}}+\negation{R_{u_q}^i}))\Big].
\end{aligned}
\end{equation}

To solve $k$-\gmca, we need to leverage off-the-shelf SAT solvers to find an
assignment of the variables such that formulas in (\ref{eq:one}), (\ref{eq:two})
and (\ref{eq:three}) are 
satisfied. For any vertex in the compatibility graph, we create $k$ variables. For
each \auxtext, we need to create $k(k+m)$ variables. Suppose that
there are $t$ vertices in the input filter $F$ and $x$ \auxestext.
Then, we need $tk+xk(k+m)$ variables. Similarly, (\ref{eq:one}) gives us $1$ clause for
each vertex, (\ref{eq:two}) gives us at most $k\times 2^t$ clauses in total,
(\ref{eq:three}) gives $k(k+m+nk+m)$ for each \auxtext. Hence, the number of
clauses to solve a $k$-\gmca is $t+k\times 2^t+xk((n+1)k+2m)$.

To find the minimum solution for \gmca, we will solve $k$-\gmca with $k$ equals
the number of states in the input filter, and the decrease $k$ until we cannot
find a solution for $k$-\gmca (i.e., the SAT solver times out).

	\else
	We follow the standard practice for treating optimization problems via a
    decision problem oracle, \emph{viz.} define a $k$-\mcca ($k$-\gmca) problem,
    asking for the existence of a cover with size $k$ satisfying the
    \auxestext; one then decreases $k$ to find the minimum cover. Each
    $k$-\mcca ($k$-\gmca) problem can be written as a logic formula in conjunctive normal
    form (CNF), polynomial in the size of the $k$-\mcca ($k$-\gmca) instance,
    and solved.  Detailed explanation of the CNF generation from the
	$k$-\mcca ($k$-\gmca) problem must be deferred to the extended version~\cite{zhang2020Cover}.
\fi

\shortvspace{-12pt}
\section{Experimental results}
\shortvspace{-6pt}

The method described was implemented by leveraging a Python implementation of
2018 SAT Competition winner,
MapleLCMDistChronoBT\,\cite{nadel2018maple,imms-sat18}. Their solver will
return a solution if it solves the $k$-\mcca problem before timing out. If it
finds a satisfying assignment, we decrease $k$, and try again.  Once no
satisfying assignment can be found, we construct a minimal filter from the
solution with minimum $k$.

First, as a sanity check, we examined the minimization problems for the inputs
shown in Figure~\ref{fig:motivation_input}, Figure~\ref{fig:input_nd} and
Figure~\ref{fig:split_choose_input}. Our implementation takes
\SI{0.05}{\second}, \SI{0.06}{\second} and \SI{0.26}{\second}, respectively, to
reduce those filters. All filters found have exactly the
minimum number of states, as reported above.

Next, designed a test where we could examine scalability aspects of the method.
We generalized the input filter shown in Figure~\ref{fig:split_choose_input} to produce a
family of instances, each described by two parameters: the $n\times m$-input
filter has $n$ rows and $m$ states at each row.  (Figure~\ref{fig:split_choose_input} is
the $2\times 3$ version.) Just like the original filter, the states in the same
row share the same color, but the states in different rows have different
colors. The initial state $w_0$ outputs a single unique color; the last two
states, $w_4$ and $w_5$, output any of the $n$ colors.  In this example, the
states in the same row, together with $w_4$ and $w_5$, are compatible with each
other.  

\begin{figure}[ht!]
\setlength{\belowcaptionskip}{2pt}
\setlength{\abovecaptionskip}{2pt}
\begin{subfigure}[b]{0.27\textwidth}
\setlength{\abovecaptionskip}{2pt}
\setlength{\belowcaptionskip}{2pt}
\centering
\includegraphics[scale=0.37]{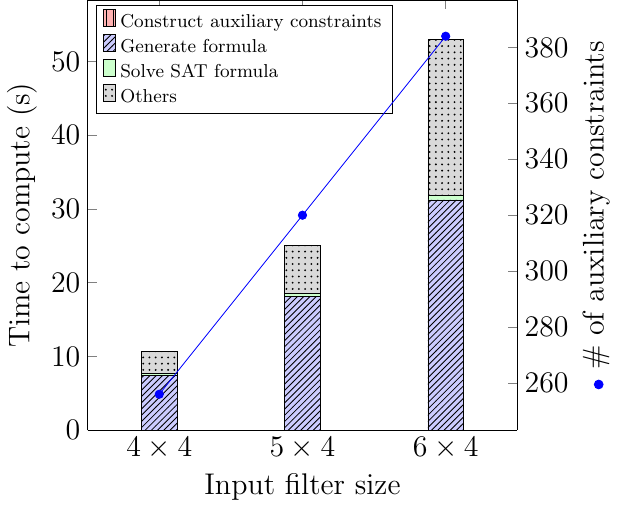}
\caption{Minimizing inputs with $4$ vertices per row.\label{fig:timing_split_4}}
\end{subfigure}
\hspace{0.45cm}
\begin{subfigure}[b]{0.27\textwidth}
\setlength{\abovecaptionskip}{2pt}
\setlength{\belowcaptionskip}{2pt}
\centering
\includegraphics[scale=0.37]{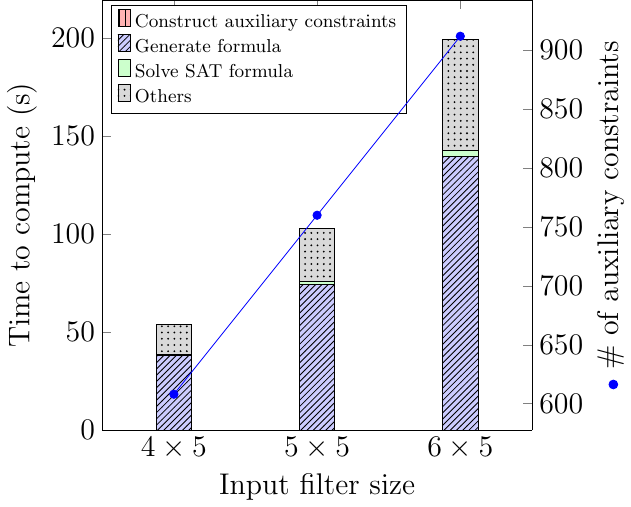}
\caption{Minimizing inputs with $5$ vertices per row.\label{fig:timing_split_5}}
\end{subfigure}
\hspace{0.55cm}
\begin{subfigure}[b]{0.27\textwidth}
\setlength{\abovecaptionskip}{2pt}
\setlength{\belowcaptionskip}{2pt}
\centering
\includegraphics[scale=0.37]{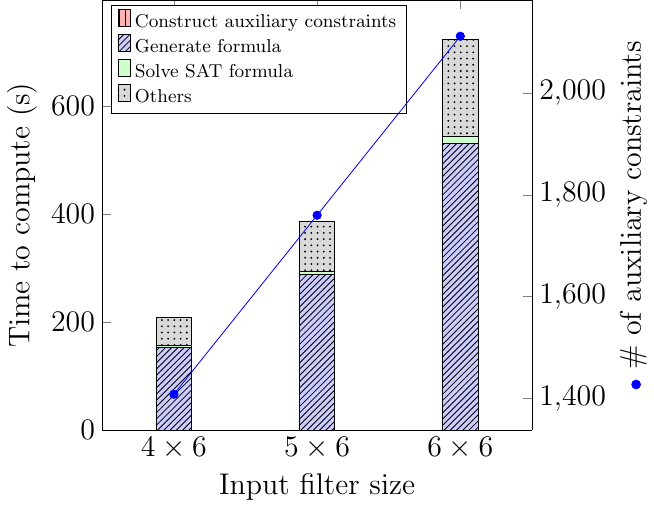}
\caption{Minimizing inputs with $6$ vertices per row.\label{fig:timing_split_6}}
\end{subfigure}
\hfill
\caption{A scalability test on the multi-outputting filter minimization
problems. Inputs are $n\times m$ parameterized instances akin to
Figure~\ref{fig:split_choose_input}.\label{fig:split}} 
\end{figure}

The time to construct the \auxestext, prepare the formulas and the
time used by the SAT solver were recorded. We also measured the number of
\auxestext found by our algorithm.  Figure~\ref{fig:split} summarizes
the data for $(n,m) \in \{4,5,6\}\times\{4,5,6\}$.  The result shows that about
\SI{70}{\percent} of the time is used in preparing the logical
formula, with the SAT solver and construction of the \auxestext accounting
for only a very small fraction of time.  


\begin{figure}
\shortvspace{-12pt}
\setlength{\belowcaptionskip}{2pt}
\setlength{\abovecaptionskip}{2pt}
\centering
\begin{subfigure}[b]{0.42\textwidth}
\setlength{\abovecaptionskip}{2pt}
\centering
\includegraphics[scale=0.9]{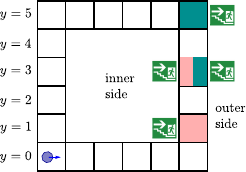}
\caption{A square environment for the robot to track the exits with the
observation indicating its current row.\label{fig:grid_scenario}}
\end{subfigure}
\hspace{0.5cm}
\begin{subfigure}[b]{0.42\textwidth}
\setlength{\abovecaptionskip}{2pt}
\centering
\includegraphics[scale=0.38]{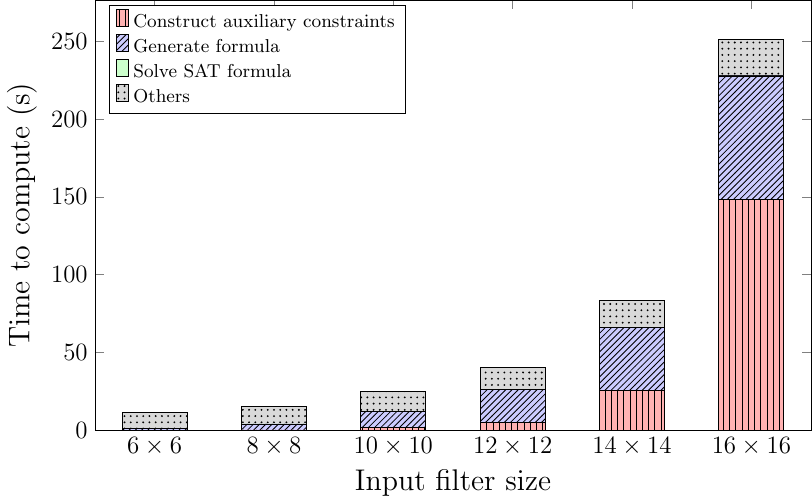}
\caption{Compute times to minimize filters with different sizes.\label{fig:grid_result}}
\end{subfigure}
\caption{Navigating to exit a square environment with a row sensor.}
\shortvspace{-18pt}
\end{figure}

In light of this, to further dissect the computational costs of different
phases, we tested a robot in the square grid environment shown in
Figure~\ref{fig:grid_scenario}. The robot starts from the bottom left cell, and
moves to some adjacent cell at each time step.  The robot only receives
observations indicating its row number at each step. We are interested in small
filter allowing the robot to recognize whether it has reached a cell with an
exit (at the inner side or outer side).  States with both inner and outer exits
have multiple outputs.  To search for a minimal filter, we firstly start with
deterministic input filters for a grid world with size $6\times 6$, $8\times
8$, $10\times 10$, $12\times 12$, $14\times 14$, $16\times 16$, and then
minimize these filters.  We collected the total time spent in different stages
of filter minimization, including the construction of \auxestext, SAT formula
generation and resolution of SAT formula by the SAT solver.  The results are
summarized visually in Figure~\ref{fig:grid_result}.

In this problem, the number of states in the input filter scales linearly with
the size of the square.  So does the minimal filter.  But the particular problem
has an important additional property: it represents a worst-case in a certain
sense because there are no \auxestext.
We do not indicate this fact to the algorithm, so the construction of \auxestext
examines many cliques, determining that none apply.  The results highlight that
the construction of the \auxestext quickly grows to overtake the time to
generate the logical formula\,---\,even though, in this case, the \auxtext set
is empty.

The preceding hints toward our direction of current research: the construction
of $\auxes{F}$ by na\"\i vely following Definition~\ref{defn:aux_const}
is costly. And, though the SAT formula is polynomial in the size of the \mcca
instance, that instance can be very large. On the other hand, the need for a
\auxtext can be detected when the output produced fails to be
deterministic. Hence, our future work will look at how to generate these
constraints lazily.

\section{Conclusion}

With an eye toward generalizing combinatorial filters, we introduced a
new class of filter, the cover filters. Then, in order to reduce the state
complexity of such filters, we re-examined earlier treatments of traditional
filter minimization; this paper has shown some prior ideas to be mistaken.
Building on these insights, we formulate the minimization problem via
compatibility graphs, examining covers comprised of cliques formed thereon. We
present an exact algorithm that generalizes from the traditional filter
minimization problem to cover filters elegantly.



\bibliographystyle{IEEEtran}
\shortenXor{\bibliography{mybib}}{\bibliography{mybibshort}}
\end{document}